\documentclass[12pt,reqno]{amsart}
\usepackage{graphicx}
\usepackage{amscd}
\usepackage{amsmath}
\usepackage{epsfig}
\usepackage{amsfonts}
\usepackage{amssymb}

\providecommand{\U}[1]{\protect\rule{.1in}{.1in}}
\providecommand{\U}[1]{\protect\rule{.1in}{.1in}}
\allowdisplaybreaks

\theoremstyle{plain}

\newtheorem{lemma}{Lemma}

\newtheorem{remark}{Remark}

\numberwithin{equation}{section}

\input{tcilatex}

\begin{document}
\title[Exact Wave Functions]{Exact Wave Functions for Generalized Harmonic
Oscillators}
\author{Nathan Lanfear}
\address{School of Mathematical and Statistical Sciences, Arizona State
University, Tempe, AZ 85287--1804, U.S.A.}
\email{nlanfear@asu.edu}
\author{Raquel M. L\'opez}
\address{Mathematical, Computational and Modeling Sciences Center, Arizona
State University, Tempe, AZ 85287--1904, U.S.A.}
\email{rlopez14@asu.edu}
\author{Sergei K. Suslov}
\address{School of Mathematical and Statistical Sciences \& Mathematical,
Computational and Modeling Sciences Center, Arizona State University, Tempe,
AZ 85287--1804, U.S.A.}
\email{sks@asu.edu}
\urladdr{http://hahn.la.asu.edu/\symbol{126}suslov/index.html}
\date{\today }
\subjclass{Primary 81Q05, 35C05. Secondary 42A38}
\keywords{The time-dependent Schr\"{o}dinger equation, generalized harmonic
oscillators, Green's function, propagator, Ermakov-type invariant and
Pinney-type solution, Ehrenfest theorem, Arnold transformation.}

\begin{abstract}
We transform the time-dependent Schr\"{o}dinger equation for the most
general variable quadratic Hamiltonians into a standard autonomous form. As
a result, the time-evolution of exact wave functions of generalized harmonic
oscillators is determined in terms of solutions of certain Ermakov and
Riccati-type systems. In addition, we show that the classical Arnold
transformation is naturally connected with Ehrenfest's theorem for
generalized harmonic oscillators.
\end{abstract}

\maketitle

\section{Introduction}

Quantum systems with variable quadratic Hamiltonians are called the
generalized harmonic oscillators (see \cite{Berry85}, \cite%
{Cor-Sot:Sua:SusInv}, \cite{Dod:Mal:Man75}, \cite{Dodonov:Man'koFIAN87}, 
\cite{Faddeyev69}, \cite{Fey:Hib}, \cite{Hannay85}, \cite{Leach90}, \cite%
{Lo93}, \cite{Malkin:Man'ko79}, \cite{Menouar:Maamache:Choi10}, \cite{Wolf81}%
, \cite{Yeon:Lee:Um:George:Pandey93}, \cite{Zhukov99} and references
therein). These systems have attracted substantial attention over the years
because of their great importance in many advanced quantum problems.
Examples are coherent states and uncertainty relations, Berry's phase,
quantization of mechanical systems and Hamiltonian cosmology. More
applications include, but are not limited to charged particle traps and
motion in uniform magnetic fields, molecular spectroscopy and polyatomic
molecules in varying external fields, crystals through which an electron is
passing and exciting the oscillator modes, and other mode interactions with
external fields. Quadratic Hamiltonians have particular applications in
quantum electrodynamics because the electromagnetic field can be represented
as a set of forced harmonic oscillators \cite{Fey:Hib}.

A goal of this Letter is to construct exact wave functions for generalized
(driven) harmonic oscillators \cite{Berry85}, \cite{Cor-Sot:Lop:Sua:Sus}, 
\cite{Hannay85}, \cite{Leach90}, \cite{Lo93}, \cite{Wolf81}, \cite%
{Yeon:Lee:Um:George:Pandey93}, in terms of Hermite polynomials by
transforming the time-dependent Schr\"{o}dinger equation into an autonomous
form \cite{Zhukov99}. The relationship with certain Ermakov and Riccati-type
systems, which seem are missing in the available literature in general, are
investigated. A group theoretical approach to a similar class of partial
differential equations is discussed in Refs.~\cite{AblowClark91}, \cite%
{Clark88}, \cite{Craddock09}, \cite{GagWint93}, \cite{Miller77}, \cite%
{Rosen76} (see also \cite{SuazoSusVega10}, \cite{SuazoSusVega11} and
references therein). Some applications to the nonlinear Schr\"{o}dinger
equation can also be found in Refs.~\cite{Dal:Giorg:Pitaevski:Str99}, \cite%
{Kagan:Surkov:Shlyap96}, \cite{Kagan:Surkov:Shlyap97}, \cite%
{Kivsh:Alex:Tur01}, \cite{Kundu09}, \cite{Per-G:Tor:Mont}, \cite%
{SuazoSuslovSol} and \cite{Suslov11}.

\section{Transforming Generalized Harmonic Oscillators into Autonomous Form}

We consider the one-dimensional time-dependent Schr\"{o}dinger equation%
\begin{equation}
i\frac{\partial \psi }{\partial t}=H\psi ,  \label{Schroedinger}
\end{equation}%
where the variable Hamiltonian $H=Q\left( p,x\right) $ is an arbitrary
quadratic of two operators $p=-i\partial /\partial x$ and $x,$ namely,%
\begin{equation}
i\psi _{t}=-a\left( t\right) \psi _{xx}+b\left( t\right) x^{2}\psi -ic\left(
t\right) x\psi _{x}-id\left( t\right) \psi -f\left( t\right) x\psi +ig\left(
t\right) \psi _{x},  \label{SchroedingerQuadratic}
\end{equation}%
($a,$ $b,$ $c,$ $d,$ $f$ and $g$ are suitable real-valued functions of time
only). We shall refer to these quantum systems as the \textit{generalized
(driven) harmonic oscillators}. Some examples, a general approach and known
elementary solutions can be found in Refs.~\cite{Cor-Sot:Lop:Sua:Sus}, \cite%
{Cor-Sot:Sua:Sus}, \cite{Cor-Sot:Sua:SusInv}, \cite{Cor-Sot:Sus}, \cite%
{Dod:Mal:Man75}, \cite{FeynmanPhD}, \cite{Feynman}, \cite{Fey:Hib}, \cite%
{Lo93}, \cite{Lop:Sus}, \cite{Me:Co:Su}, \cite{Suaz:Sus}, \cite{Wolf81} and 
\cite{Yeon:Lee:Um:George:Pandey93}. In addition, a case related to Airy
functions is discussed in \cite{Lan:Sus} and Ref.~\cite{Cor-Sot:SusDPO}
deals with another special case of transcendental solutions.

The following is our first result.

\begin{lemma}
The substitution%
\begin{equation}
\psi =\frac{e^{i\left( \alpha \left( t\right) x^{2}+\delta \left( t\right)
x+\kappa \left( t\right) \right) }}{\sqrt{\mu \left( t\right) }}\ \chi
\left( \xi ,\tau \right) ,\qquad \xi =\beta \left( t\right) x+\varepsilon
\left( t\right) ,\quad \tau =\gamma \left( t\right)  \label{Ansatz}
\end{equation}%
transforms the non-autonomous and inhomogeneous Schr\"{o}dinger equation (%
\ref{SchroedingerQuadratic}) into the autonomous form%
\begin{equation}
-i\chi _{\tau }=-\chi _{\xi \xi }+c_{0}\xi ^{2}\chi \qquad \left(
c_{0}=0,1\right)  \label{ASEq}
\end{equation}%
provided that%
\begin{equation}
\frac{d\alpha }{dt}+b+2c\alpha +4a\alpha ^{2}=c_{0}a\beta ^{4},  \label{SysA}
\end{equation}%
\begin{equation}
\frac{d\beta }{dt}+\left( c+4a\alpha \right) \beta =0,  \label{SysB}
\end{equation}%
\begin{equation}
\frac{d\gamma }{dt}+a\beta ^{2}=0  \label{SysC}
\end{equation}%
and%
\begin{equation}
\frac{d\delta }{dt}+\left( c+4a\alpha \right) \delta =f+2g\alpha
+2c_{0}a\beta ^{3}\varepsilon ,  \label{SysD}
\end{equation}%
\begin{equation}
\frac{d\varepsilon }{dt}=\left( g-2a\delta \right) \beta ,  \label{SysE}
\end{equation}%
\begin{equation}
\frac{d\kappa }{dt}=g\delta -a\delta ^{2}+c_{0}a\beta ^{2}\varepsilon ^{2}.
\label{SysF}
\end{equation}%
Here%
\begin{equation}
\alpha =\frac{1}{4a}\frac{\mu ^{\prime }}{\mu }-\frac{d}{2a}.  \label{Alpha}
\end{equation}
\end{lemma}

\begin{proof}
Differentiating $\psi =\mu ^{-1/2}\left( t\right) e^{iS\left( x,t\right)
}\chi \left( \xi ,\tau \right) $ with $S=\alpha \left( t\right) x^{2}+\delta
\left( t\right) x+\kappa \left( t\right) ,$ $\xi =\beta \left( t\right)
x+\varepsilon \left( t\right) $ and $\tau =\gamma \left( t\right) $ yields%
\begin{equation}
ie^{-iS}\psi _{t}=\frac{1}{\sqrt{\mu }}\left[ -\left( \alpha ^{\prime
}x^{2}+\delta ^{\prime }x+\kappa ^{\prime }\right) \chi +i\left( \left(
\beta ^{\prime }x+\varepsilon ^{\prime }\right) \chi _{\xi }+\gamma ^{\prime
}\chi _{\tau }-\frac{\mu ^{\prime }}{2\mu }\chi \right) \right] ,
\label{Pf1}
\end{equation}%
\begin{equation}
e^{-iS}\psi _{x}=\frac{1}{\sqrt{\mu }}\left[ i\left( 2\alpha x+\delta
\right) \chi +\beta \chi _{\xi }\right]  \label{Pf2}
\end{equation}%
and%
\begin{equation}
e^{-iS}\psi _{xx}=\frac{1}{\sqrt{\mu }}\left[ \left( 2i\alpha -\left(
2\alpha x+\delta \right) ^{2}\right) \chi +2i\left( 2\alpha x+\delta \right)
\beta \chi _{\xi }+\beta ^{2}\chi _{\xi \xi }\right] .  \label{Pf3}
\end{equation}%
Substituting into%
\begin{eqnarray}
i\psi _{t} &=&-a\psi _{xx}+\left( b-c_{0}a\beta ^{4}\right) x^{2}\psi
-icx\psi _{x}-id\psi  \label{Pf4} \\
&&-\left( f+2c_{0}a\beta ^{3}\varepsilon \right) x\psi +ig\psi
_{x}-c_{0}a\beta ^{2}\varepsilon ^{2}\psi +c_{0}a\beta ^{2}\varepsilon
^{2}\xi ^{2}\psi  \notag
\end{eqnarray}%
and using system (\ref{SysA})--(\ref{SysF}), results in Eq.~(\ref{ASEq}).
Further computational details are left to the reader.
\end{proof}

Our transformation (\ref{Ansatz}) provides a new interpretation to system (%
\ref{SysA})--(\ref{SysF}) originally derived in Ref.~\cite%
{Cor-Sot:Lop:Sua:Sus} when $c_{0}=0$ by integrating the corresponding Schr%
\"{o}dinger equation via the Green function method (see also \cite{Suslov10}
for an eigenfunction expansion). Here, we discuss the case $c_{0}\neq 0$ as
its natural extension.

The substitution (\ref{Alpha}), which has been already used in \cite%
{Cor-Sot:Lop:Sua:Sus}, appears here from a new \textquotedblleft
transformation perspective\textquotedblright . It now reduces the
inhomogeneous equation (\ref{SysA}) to the second order ordinary
differential equation%
\begin{equation}
\mu ^{\prime \prime }-\tau \left( t\right) \mu ^{\prime }+4\sigma \left(
t\right) \mu =c_{0}\left( 2a\right) ^{2}\beta ^{4}\mu ,  \label{CharEq}
\end{equation}%
that has the familiar time-varying coefficients%
\begin{equation}
\tau \left( t\right) =\frac{a^{\prime }}{a}-2c+4d,\qquad \sigma \left(
t\right) =ab-cd+d^{2}+\frac{d}{2}\left( \frac{a^{\prime }}{a}-\frac{%
d^{\prime }}{d}\right) .  \label{TauSigma}
\end{equation}%
(The reader should be convinced that this derivation is rather
straightforward.)

When $c_{0}=0$, equation (\ref{SysA}) is called the \textit{Riccati
nonlinear differential equation} \cite{Wa}, \cite{Wh:Wa}; consequently, the
system (\ref{SysA})--(\ref{SysF}) shall be referred to as a \textit{%
Riccati-type system}. (Similar terminology is used in \cite{SuazoSusVega11}
for the corresponding parabolic equation.) Now if $c_{0}=1,$\ equation (\ref%
{CharEq}) can be reduced to a generalized version of the \textit{Ermakov
nonlinear differential equation} (\ref{Ermakov}) (see, for example, \cite%
{Cor-Sot:Sua:SusInv}, \cite{Ermakov}, \cite{Leach:Andrio08}, \cite{Suslov10}
and references therein regarding Ermakov's equation) and we shall refer to
the corresponding system (\ref{SysA})--(\ref{SysF}) with $c_{0}\neq 0$ as an 
\textit{Ermakov-type system}.

\section{Green's Function and Wavefunctions}

Two particular solutions of the time-dependent Schr\"{o}dinger equation (\ref%
{SchroedingerQuadratic}) are useful in physical applications. Using standard
oscillator wave functions for equation (\ref{ASEq}) when $c_{0}=1$ (for
example, \cite{Flu}, \cite{La:Lif} and/or \cite{Merz}) results in the
solution%
\begin{equation}
\psi _{n}\left( x,t\right) =\frac{e^{i\left( \alpha x^{2}+\delta x+\kappa
\right) +i\left( 2n+1\right) \gamma }}{\sqrt{2^{n}n!\mu \sqrt{\pi }}}\
e^{-\left( \beta x+\varepsilon \right) ^{2}/2}\ H_{n}\left( \beta
x+\varepsilon \right) ,  \label{WaveFunction}
\end{equation}%
where $H_{n}\left( x\right) $ are the Hermite polynomials \cite{Ni:Su:Uv},
provided that the solution of the Ermakov-type system (\ref{SysA})--(\ref%
{SysF}) is available.

The Green function of generalized harmonic oscillators has been constructed
in the following fashion in Ref.~\cite{Cor-Sot:Lop:Sua:Sus}:%
\begin{equation}
G\left( x,y,t\right) =\frac{1}{\sqrt{2\pi i\mu _{0}\left( t\right) }}\exp %
\left[ i\left( \alpha _{0}\left( t\right) x^{2}+\beta _{0}\left( t\right)
xy+\gamma _{0}\left( t\right) y^{2}+\delta _{0}\left( t\right) x+\varepsilon
_{0}\left( t\right) y+\kappa _{0}\left( t\right) \right) \right] .
\label{GreenFunction}
\end{equation}%
The time-dependent coefficients $\alpha _{0},$ $\beta _{0},$ $\gamma _{0},$ $%
\delta _{0},$ $\varepsilon _{0},$ $\kappa _{0}$ satisfy the Riccati-type
system (\ref{SysA})--(\ref{SysF}) $(c_{0}=0)$ and are given as follows \cite%
{Cor-Sot:Lop:Sua:Sus}, \cite{Suaz:Sus}, \cite{Suslov10}:%
\begin{eqnarray}
&&\alpha _{0}\left( t\right) =\frac{1}{4a\left( t\right) }\frac{\mu
_{0}^{\prime }\left( t\right) }{\mu _{0}\left( t\right) }-\frac{d\left(
t\right) }{2a\left( t\right) },  \label{A0} \\
&&\beta _{0}\left( t\right) =-\frac{\lambda \left( t\right) }{\mu _{0}\left(
t\right) },\qquad \lambda \left( t\right) =\exp \left( -\int_{0}^{t}\left(
c\left( s\right) -2d\left( s\right) \right) \ ds\right) ,  \label{B0} \\
&&\gamma _{0}\left( t\right) =\frac{1}{2\mu _{1}\left( 0\right) }\frac{\mu
_{1}\left( t\right) }{\mu _{0}\left( t\right) }+\frac{d\left( 0\right) }{%
2a\left( 0\right) }  \label{C0}
\end{eqnarray}%
and%
\begin{equation}
\delta _{0}\left( t\right) =\frac{\lambda \left( t\right) }{\mu _{0}\left(
t\right) }\int_{0}^{t}\left[ \left( f\left( s\right) -\frac{d\left( s\right) 
}{a\left( s\right) }g\left( s\right) \right) \mu _{0}\left( s\right) +\frac{%
g\left( s\right) }{2a\left( s\right) }\mu _{0}^{\prime }\left( s\right) %
\right] \frac{ds}{\lambda \left( s\right) },  \label{D0}
\end{equation}%
\begin{eqnarray}
\varepsilon _{0}\left( t\right) &=&-\frac{2a\left( t\right) \lambda \left(
t\right) }{\mu _{0}^{\prime }\left( t\right) }\delta _{0}\left( t\right)
+8\int_{0}^{t}\frac{a\left( s\right) \sigma \left( s\right) \lambda \left(
s\right) }{\left( \mu _{0}^{\prime }\left( s\right) \right) ^{2}}\left( \mu
_{0}\left( s\right) \delta _{0}\left( s\right) \right) \ ds  \label{E0} \\
&&\quad +2\int_{0}^{t}\frac{a\left( s\right) \lambda \left( s\right) }{\mu
_{0}^{\prime }\left( s\right) }\left( f\left( s\right) -\frac{d\left(
s\right) }{a\left( s\right) }g\left( s\right) \right) \ ds,  \notag
\end{eqnarray}%
\begin{eqnarray}
\kappa _{0}\left( t\right) &=&\frac{a\left( t\right) \mu _{0}\left( t\right) 
}{\mu _{0}^{\prime }\left( t\right) }\delta _{0}^{2}\left( t\right)
-4\int_{0}^{t}\frac{a\left( s\right) \sigma \left( s\right) }{\left( \mu
_{0}^{\prime }\left( s\right) \right) ^{2}}\left( \mu _{0}\left( s\right)
\delta _{0}\left( s\right) \right) ^{2}\ ds  \label{F0} \\
&&\quad -2\int_{0}^{t}\frac{a\left( s\right) }{\mu _{0}^{\prime }\left(
s\right) }\left( \mu _{0}\left( s\right) \delta _{0}\left( s\right) \right)
\left( f\left( s\right) -\frac{d\left( s\right) }{a\left( s\right) }g\left(
s\right) \right) \ ds  \notag
\end{eqnarray}%
$(\delta _{0}\left( 0\right) =-\varepsilon _{0}\left( 0\right) =g\left(
0\right) /\left( 2a\left( 0\right) \right) $ and $\kappa _{0}\left( 0\right)
=0)$ provided that $\mu _{0}$ and $\mu _{1}$ are standard solutions of
equation (\ref{CharEq}) with $c_{0}=0$ corresponding to the initial
conditions $\mu _{0}\left( 0\right) =0,$ $\mu _{0}^{\prime }\left( 0\right)
=2a\left( 0\right) \neq 0$ and $\mu _{1}\left( 0\right) \neq 0,$ $\mu
_{1}^{\prime }\left( 0\right) =0.$ (Proofs of these facts are outlined in
Refs.~\cite{Cor-Sot:Lop:Sua:Sus}, \cite{Cor-Sot:SusDPO} and \cite{Suaz:Sus}.
See also important previous works \cite{Dodonov:Man'koFIAN87}, \cite%
{Malkin:Man'ko79}, \cite{Wolf81}, \cite{Yeon:Lee:Um:George:Pandey93}, \cite%
{Zhukov99} and references therein for more details.)

Hence, the corresponding Cauchy initial value problem can be solved
(formally) by the superposition principle:%
\begin{equation}
\psi \left( x,t\right) =\int_{-\infty }^{\infty }G\left( x,y,t\right) \psi
\left( y,0\right) \ dy  \label{Superposition}
\end{equation}%
for some suitable initial data $\psi \left( x,0\right) =\varphi \left(
x\right) $ (see Refs.~\cite{Cor-Sot:Lop:Sua:Sus}, \cite{Suaz:Sus} and \cite%
{Suslov10} for further details).

In particular, using the wave functions (\ref{WaveFunction}) we get the
integral%
\begin{equation}
\psi _{n}\left( x,t\right) =\int_{-\infty }^{\infty }G\left( x,y,t\right)
\psi _{n}\left( y,0\right) \ dy,  \label{SuperWF}
\end{equation}%
and this can be evaluated by%
\begin{eqnarray}
&&\int_{-\infty }^{\infty }e^{-\lambda ^{2}\left( x-y\right)
^{2}}H_{n}\left( ay\right) \ dy  \label{Erd} \\
&&\quad =\frac{\sqrt{\pi }}{\lambda ^{n+1}}\left( \lambda ^{2}-a^{2}\right)
^{n/2}H_{n}\left( \frac{\lambda ax}{\left( \lambda ^{2}-a^{2}\right) ^{1/2}}%
\right) ,\quad \func{Re}\lambda ^{2}>0,  \notag
\end{eqnarray}%
which is an integral transform equivalent to Eq.~(30) on page 195 of Vol.~2
of Ref.~\cite{Erd} (the Gauss transform of Hermite polynomials), or Eq.~(17)
on page 290 of Vol.~2 of Ref.~\cite{ErdInt}.

\section{Solution to Ermakov-type System}

As shown in the previous section, the time evolution of the wave functions (%
\ref{WaveFunction}) is determined in terms of the solution to the initial
value problem for the Ermakov-type system. In this section, formulas (\ref%
{WaveFunction})--(\ref{GreenFunction}) and (\ref{SuperWF})--(\ref{Erd})
shall be used in order to solve the general system (\ref{SysA})--(\ref{SysF}%
) when $c_{0}\neq 0$ along with the uniqueness property of the Cauchy
initial value problem. At this point, we must remind the reader how to
handle the special case $c_{0}=0$ considered in \cite{Suaz:Sus}.

\begin{lemma}
The solution of the Riccati-type system (\ref{SysA})--(\ref{SysF}) $%
(c_{0}=0) $ is given by%
\begin{eqnarray}
&&\mu \left( t\right) =2\mu \left( 0\right) \mu _{0}\left( t\right) \left(
\alpha \left( 0\right) +\gamma _{0}\left( t\right) \right) ,  \label{MKernel}
\\
&&\alpha \left( t\right) =\alpha _{0}\left( t\right) -\frac{\beta
_{0}^{2}\left( t\right) }{4\left( \alpha \left( 0\right) +\gamma _{0}\left(
t\right) \right) },  \label{AKernel} \\
&&\beta \left( t\right) =-\frac{\beta \left( 0\right) \beta _{0}\left(
t\right) }{2\left( \alpha \left( 0\right) +\gamma _{0}\left( t\right)
\right) }=\frac{\beta \left( 0\right) \mu \left( 0\right) }{\mu \left(
t\right) }\lambda \left( t\right) ,  \label{BKernel} \\
&&\gamma \left( t\right) =\gamma \left( 0\right) -\frac{\beta ^{2}\left(
0\right) }{4\left( \alpha \left( 0\right) +\gamma _{0}\left( t\right)
\right) }  \label{CKernel}
\end{eqnarray}%
and%
\begin{eqnarray}
\delta \left( t\right) &=&\delta _{0}\left( t\right) -\frac{\beta _{0}\left(
t\right) \left( \delta \left( 0\right) +\varepsilon _{0}\left( t\right)
\right) }{2\left( \alpha \left( 0\right) +\gamma _{0}\left( t\right) \right) 
},  \label{DKernel} \\
\varepsilon \left( t\right) &=&\varepsilon \left( 0\right) -\frac{\beta
\left( 0\right) \left( \delta \left( 0\right) +\varepsilon _{0}\left(
t\right) \right) }{2\left( \alpha \left( 0\right) +\gamma _{0}\left(
t\right) \right) },  \label{EKernel} \\
\kappa \left( t\right) &=&\kappa \left( 0\right) +\kappa _{0}\left( t\right)
-\frac{\left( \delta \left( 0\right) +\varepsilon _{0}\left( t\right)
\right) ^{2}}{4\left( \alpha \left( 0\right) +\gamma _{0}\left( t\right)
\right) }  \label{FKernel}
\end{eqnarray}%
in terms of the fundamental solution (\ref{A0})--(\ref{F0}) subject to the
arbitrary initial data $\mu \left( 0\right) ,$ $\alpha \left( 0\right) ,$ $%
\beta \left( 0\right) \neq 0,$ $\gamma \left( 0\right) ,$ $\delta \left(
0\right) ,$ $\varepsilon \left( 0\right) ,$ $\kappa \left( 0\right)$.
\end{lemma}

This solution can be verified by a direct substitution and/or by an integral
evaluation. This result can also be thought of as a nonlinear superposition
principle for the Riccati-type system and the continuity with respect to
initial data holds \cite{Suaz:Sus}.

Hence, the solution (\ref{MKernel})--(\ref{FKernel}) implies the following
asymptotics established in \cite{Suaz:Sus}:%
\begin{eqnarray}
&&\alpha _{0}\left( t\right) =\frac{1}{4a\left( 0\right) t}-\frac{c\left(
0\right) }{4a\left( 0\right) }-\frac{a^{\prime }\left( 0\right) }{%
8a^{2}\left( 0\right) }+\mathcal{O}\left( t\right) ,  \label{AssA0} \\
&&\beta _{0}\left( t\right) =-\frac{1}{2a\left( 0\right) t}+\frac{a^{\prime
}\left( 0\right) }{4a^{2}\left( 0\right) }+\mathcal{O}\left( t\right) ,
\label{AssB0} \\
&&\gamma _{0}\left( t\right) =\frac{1}{4a\left( 0\right) t}+\frac{c\left(
0\right) }{4a\left( 0\right) }-\frac{a^{\prime }\left( 0\right) }{%
8a^{2}\left( 0\right) }+\mathcal{O}\left( t\right) ,  \label{AssC0} \\
&&\delta _{0}\left( t\right) =\frac{g\left( 0\right) }{2a\left( 0\right) }+%
\mathcal{O}\left( t\right) ,\quad \varepsilon _{0}\left( t\right) =-\frac{%
g\left( 0\right) }{2a\left( 0\right) }+\mathcal{O}\left( t\right) ,
\label{AssD0} \\
&&\kappa _{0}\left( t\right) =\mathcal{O}\left( t\right)  \label{AssE0}
\end{eqnarray}%
as $t\rightarrow 0$ for sufficiently smooth coefficients of the original Schr%
\"{o}dinger equation (\ref{SchroedingerQuadratic}). Therefore,%
\begin{eqnarray}
G\left( x,y,t\right) &\sim &\frac{1}{\sqrt{2\pi ia\left( 0\right) t}}\exp %
\left[ i\frac{\left( x-y\right) ^{2}}{4a\left( 0\right) t}\right]
\label{GreenAsymp} \\
&&\times \exp \left[ -i\left( \frac{a^{\prime }\left( 0\right) }{%
8a^{2}\left( 0\right) }\left( x-y\right) ^{2}+\frac{c\left( 0\right) }{%
4a\left( 0\right) }\left( x^{2}-y^{2}\right) -\frac{g\left( 0\right) }{%
2a\left( 0\right) }\left( x-y\right) \right) \right]  \notag
\end{eqnarray}%
as $t\rightarrow 0$ (where $f\sim g$ as $t\rightarrow 0,$ if $%
\lim_{t\rightarrow 0}\left( f/g\right) =$ $1$). This corrects an errata in
Ref.~\cite{Cor-Sot:Lop:Sua:Sus} .

Finally, we present the extension to a general case when $c_{0}\neq 0.$ Our
main result is the following.

\begin{lemma}
The solution of the Ermakov-type system (\ref{SysA})--(\ref{SysF}) when $%
c_{0}=1\left( \neq 0\right) $ is given by%
\begin{eqnarray}
&&\mu =\mu \left( 0\right) \mu _{0}\sqrt{\beta ^{4}\left( 0\right) +4\left(
\alpha \left( 0\right) +\gamma _{0}\right) ^{2}},  \label{MKernelOsc} \\
&&\alpha =\alpha _{0}-\beta _{0}^{2}\frac{\alpha \left( 0\right) +\gamma _{0}%
}{\beta ^{4}\left( 0\right) +4\left( \alpha \left( 0\right) +\gamma
_{0}\right) ^{2}},  \label{AKernelOsc} \\
&&\beta =-\frac{\beta \left( 0\right) \beta _{0}}{\sqrt{\beta ^{4}\left(
0\right) +4\left( \alpha \left( 0\right) +\gamma _{0}\right) ^{2}}}=\frac{%
\beta \left( 0\right) \mu \left( 0\right) }{\mu \left( t\right) }\lambda
\left( t\right) ,  \label{BKernelOsc} \\
&&\gamma =\gamma \left( 0\right) -\frac{1}{2}\arctan \frac{\beta ^{2}\left(
0\right) }{2\left( \alpha \left( 0\right) +\gamma _{0}\right) },\quad
a\left( 0\right) >0  \label{CKernelOsc}
\end{eqnarray}%
and%
\begin{eqnarray}
\delta &=&\delta _{0}-\beta _{0}\frac{\varepsilon \left( 0\right) \beta
^{3}\left( 0\right) +2\left( \alpha \left( 0\right) +\gamma _{0}\right)
\left( \delta \left( 0\right) +\varepsilon _{0}\right) }{\beta ^{4}\left(
0\right) +4\left( \alpha \left( 0\right) +\gamma _{0}\right) ^{2}},
\label{DKernelOsc} \\
\varepsilon &=&\frac{2\varepsilon \left( 0\right) \left( \alpha \left(
0\right) +\gamma _{0}\right) -\beta \left( 0\right) \left( \delta \left(
0\right) +\varepsilon _{0}\right) }{\sqrt{\beta ^{4}\left( 0\right) +4\left(
\alpha \left( 0\right) +\gamma _{0}\right) ^{2}}},  \label{EKernelOsc} \\
\kappa &=&\kappa \left( 0\right) +\kappa _{0}-\varepsilon \left( 0\right)
\beta ^{3}\left( 0\right) \frac{\delta \left( 0\right) +\varepsilon _{0}}{%
\beta ^{4}\left( 0\right) +4\left( \alpha \left( 0\right) +\gamma
_{0}\right) ^{2}}  \label{FKernelOsc} \\
&&+\left( \alpha \left( 0\right) +\gamma _{0}\right) \frac{\varepsilon
^{2}\left( 0\right) \beta ^{2}\left( 0\right) -\left( \delta \left( 0\right)
+\varepsilon _{0}\right) ^{2}}{\beta ^{4}\left( 0\right) +4\left( \alpha
\left( 0\right) +\gamma _{0}\right) ^{2}}  \notag
\end{eqnarray}%
in terms of the fundamental solution (\ref{A0})--(\ref{F0}) subject to the
arbitrary initial data $\mu \left( 0\right) ,$ $\alpha \left( 0\right) ,$ $%
\beta \left( 0\right) \neq 0,$ $\gamma \left( 0\right) ,$ $\delta \left(
0\right) ,$ $\varepsilon \left( 0\right) ,$ $\kappa \left( 0\right) .$
\end{lemma}

Following are the steps to the sketch of the proof. Evaluate the integral (%
\ref{SuperWF}) with the help of (\ref{Erd}) by completing the square and
simplify. Use the uniqueness property of the Cauchy initial value problem.
One can also verify our solution by a direct substitution into the system (%
\ref{SysA})--(\ref{SysF}) when $c_{0}=1.$ These elementary but rather
tedious calculations are left to the reader (the use of a computer algebra
system is helpful at certain steps).

Furthermore, the asymptotics (\ref{AssA0})--(\ref{AssE0}) together with our
formulas (\ref{MKernelOsc})--(\ref{FKernelOsc}) result in the continuity
with respect to initial data:%
\begin{equation}
\lim_{t\rightarrow 0^{+}}\mu \left( t\right) =\mu \left( 0\right) ,\quad
\lim_{t\rightarrow 0^{+}}\alpha \left( t\right) =\alpha \left( 0\right)
,\quad \text{etc.}  \label{lims}
\end{equation}%
Thus the transformation property (\ref{MKernelOsc})--(\ref{FKernelOsc})
allows us to find a solution of the initial value problem in terms of the
fundamental solution (\ref{A0})--(\ref{F0}) and it may be referred to as a 
\textit{nonlinear superposition principle} for the Ermakov-type system.

\section{Solution of the Ermakov-type Equation}

Starting from (\ref{CharEq})--(\ref{TauSigma}) when $c_{0}=1,$ and using (%
\ref{BKernelOsc}) we arrive at%
\begin{equation}
\mu ^{\prime \prime }-\tau \left( t\right) \mu ^{\prime }+4\sigma \left(
t\right) \mu =\left( 2a\right) ^{2}\left( \beta \left( 0\right) \mu \left(
0\right) \lambda \right) ^{4}\mu ^{-3},  \label{Ermakov}
\end{equation}%
which is a familiar Ermakov-type equation (see \cite{Car:Luc08}, \cite%
{Cor-Sot:Sua:SusInv}, \cite{Ermakov}, \cite{Leach:Andrio08}, \cite{Suslov10}%
, \cite{Zhukov99} and references therein). Then our formula (\ref{MKernelOsc}%
) leads to the representation%
\begin{equation}
\left( \frac{\mu \left( t\right) }{\mu \left( 0\right) }\right) ^{2}=\beta
^{4}\left( 0\right) \mu _{0}^{2}\left( t\right) +\left( \frac{\mu _{1}\left(
t\right) }{\mu _{1}\left( 0\right) }+\frac{\mu ^{\prime }\left( 0\right) }{%
2\mu \left( 0\right) }\frac{\mu _{0}\left( t\right) }{a\left( 0\right) }%
\right) ^{2}  \label{Pinney}
\end{equation}%
given in terms of standard solutions $\mu _{0}$ and $\mu _{1}$ of the linear
characteristic equation (\ref{CharEq}) when $c_{0}=0.$ Further details on
this Pinney-type solution and the corresponding Ermakov-type invariant are
left to the reader (see also \cite{Car:Luc08} and \cite{Suslov10}).

\section{Ehrenfest Theorem Transformations}

By introducing an expectation value of the coordinate operator in the
following form%
\begin{equation}
\overline{x}=\frac{\left\langle x\right\rangle }{\left\langle 1\right\rangle 
}=\frac{\left\langle \psi ,x\psi \right\rangle }{\left\langle \psi ,\psi
\right\rangle },  \label{ExpX}
\end{equation}%
one can derive Ehrenfest's theorem for the generalized (driven) harmonic
oscillators (see, for example, \cite{Cor-Sot:Sua:Sus} and \cite%
{Cor-Sot:Sua:SusInv}). Then the following classical equation of motion of
the parametric driven oscillator holds%
\begin{equation}
\frac{d^{2}\overline{x}}{dt^{2}}-\frac{a^{\prime }}{a}\frac{d\overline{x}}{dt%
}+\left( 4ab-c^{2}+c\frac{a^{\prime }}{a}-c^{\prime }\right) \overline{x}%
=2af-g^{\prime }+g\frac{a^{\prime }}{a}-cg.  \label{Ehrenfest}
\end{equation}

The transformation of the expectation values%
\begin{equation}
\overline{\xi }=\beta \ \overline{x}+\varepsilon ,\qquad \overline{\xi }%
=\left\langle \chi ,\xi \chi \right\rangle \quad \text{with\quad }%
\left\langle \chi ,\chi \right\rangle =1,  \label{ArnoldTransform}
\end{equation}%
corresponding to our Lemma~1, converts (\ref{Ehrenfest}) into the simplest
equation of motion of the free particle and/or harmonic oscillator:%
\begin{equation}
\frac{d^{2}\overline{\xi }}{d\tau ^{2}}+4c_{0}\overline{\xi }=0\qquad \left(
c_{0}=0,1\right) .  \label{harmonic}
\end{equation}%
(This can be verified by a direct calculation.)

\begin{remark}
An exact transformation of a linear second-order differential equation into
the equation of motion of free particle was discussed by Arnold \cite%
{Arnold98}. An extension of the later to the case of the time-dependent Schr%
\"{o}dinger equation had been considered, for example, in Ref.~\cite%
{Zhukov99} and recently it has been reproduced as the quantum Arnold
transformation in \cite{Ald:Coss:Guerr:Lop-Ru11} and \cite%
{Guerr:Lop:Ald:Coss11} (see also \cite{AblowClark91}, \cite{Clark88}, \cite%
{Craddock09}, \cite{GagWint93}, \cite{Kundu09}, \cite{Miller77}, \cite%
{Rosen76} and \cite{Suslov11} for similar transformations of nonlinear Schr%
\"{o}dinger and other equations of mathematical physics). We elaborate on a
relation of the quantum Arnold transformation for the generalized (driven)
harmonic oscillators with a Riccati-type system when $c_{0}=0$
(transformation to the free particle) and consider an extension of this
transformation (in terms of solutions of the corresponding Ermakov-type
system) to the case $c_{0}=1$ (transformation to the classical harmonic
oscillator \cite{Zhukov99}).
\end{remark}

\section{Conclusion}

In this Letter, we have determined the time evolution of the wave functions
of generalized (driven) harmonic oscillators (\ref{WaveFunction}), known for
their great importance in many advanced quantum problems \cite{Fey:Hib}, in
terms of the solution to the Ermakov-type system (\ref{SysA})--(\ref{SysF})
by means of a variant of the nonlinear superposition principle (\ref%
{MKernelOsc})--(\ref{FKernelOsc}). Moreover, the classical Arnold
transformation is related to Ehrenfest's theorem. Numerous examples, the
corresponding coherent states, dynamic invariants, eigenfunction expansions
and transition amplitudes \cite{Dodonov:Man'koFIAN87}, \cite{Lan:Sus}, \cite%
{Leach90}, \cite{Lo93}, \cite{Malkin:Man'ko79}, \cite{Malk:Man:Trif73}, \cite%
{Suslov10} will be discussed elsewhere.

\noindent \textbf{Acknowledgments.\/} We thank Professor Carlos Castillo-Ch%
\'{a}vez and Professor Vladimir~I.~Man'ko for support, valuable discussions
and encouragement. The authors are indebted to Professor Francisco F.~L\'{o}%
pez-Ruiz for kindly pointing out the papers \cite{Ald:Coss:Guerr:Lop-Ru11}
and \cite{Guerr:Lop:Ald:Coss11} to our attention and for valuable
discussions. We are grateful to the organizers of the $12^{\text{th}}$
ICSSUR (Foz do Igua\c{c}u, Brazil, May 02--06, 2011) for their hospitality
and an opportunity to present the results of this work. This research is
supported in part by the National Science Foundation--Enhancing the
Mathematical Sciences Workforce in the 21st Century (EMSW21), award \#
0838705; the Alfred P. Sloan Foundation--Sloan National Pipeline Program in
the Mathematical and Statistical Sciences, award \# LTR 05/19/09; and the
National Security Agency--Mathematical \& Theoretical Biology
Institute---Research program for Undergraduates; award \# H98230-09-1-0104.

\end{document}